\documentclass[a4paper]{IEEEtran}

\usepackage{amsmath,amssymb,amsthm,braket,dsfont,subfig,graphicx,hyperref,mathtools,qcircuit,enumerate}
\usepackage[noadjust]{cite}

\newcommand{\cmultigate}[2]{*+<1em,.9em>{\hphantom{#2}} \POS [0,0]="i",[0,0].[#1,0]="e",!C *{#2},"e"+UR;"e"+UL **\dir{-};"e"+DL **\dir{-};"e"+DR **\dir{-};"e"+UR **\dir{-},"i" \cw}

\newtheorem{dfn}{Definition}
\newtheorem{lmm}{Lemma}
\newtheorem{thm}{Theorem}
\newtheorem{cor}{Corollary}

\DeclareMathOperator{\Tr}{Tr}
\DeclareMathOperator{\openone}{\mathds{1}}
\DeclareMathOperator{\inst}{inst}
\DeclareMathOperator{\gram}{gram}
\DeclareMathOperator{\bv}{bv}
\DeclareMathOperator*{\argmax}{arg\,max}

\DeclarePairedDelimiter{\ceil}{\lceil}{\rceil}
\DeclarePairedDelimiter{\floor}{\lfloor}{\rfloor}

\begin{document}

\title{On the Measurement attaining the Quantum Guesswork}

\author{Michele  Dall'Arno\thanks{M. Dall'Arno  is with  the
    Department   of   Computer  Science   and   Engineering,
    Toyohashi  University  of Technology,  1-1  Hibarigaoka,
    Tempaku-cho, Toyohashi, Aichi, 441-8580, Japan, and with
    the  Yukawa  Institute  for Theoretical  Physics,  Kyoto
    University, Sakyo-ku,  Kyoto, 606-8502,  Japan (e--mail:
    michele.dallarno.mv@tut.jp)}}

\maketitle

\begin{abstract}
  The  guesswork quantifies  the  minimum  cost incurred  in
  guessing the state of an ensemble, when only one state can
  be queried  at a time. In  the classical case, it  is well
  known  that the  optimal  strategy  trivially consists  of
  querying  the  states  in their  non-increasing  order  of
  posterior probability.  In the  quantum case, on the other
  hand,  the most  general  strategy to  obtain the  optimal
  ordering  in which  to perform  the queries  consist of  a
  quantum measurement.  Here, we  solve such an optimization
  problem by deriving the  quantum measurement attaining the
  guesswork  for  a  broad   class  of  ensembles  and  cost
  functions.
\end{abstract}

In the classical scenario, the guesswork problem, introduced
by Massey~\cite{Mas94} almost three decades ago, consists of
minimizing the number  of queries needed to  guess the value
of a random variable, when only  one value can be queried at
a time.   Already in Massey's  seminal work, it  was noticed
that the optimal strategy trivially consists of querying the
values in non-increasing order of their (prior) probability.
Hence, Massey's  work, as  well as  following works  such as
those by Arikan~\cite{Ari96}  and others~\cite{AM98a, AM98b,
  Pli98, MS04, Sun07, HS10, CD13, SV18, Sas18, AH20}, mostly
focused  on characterizing  properties  of  such an  optimal
strategy, in particular in relation with the Renyi entropies
of such a probability distribution.

Even  in   the  presence   of  classical   side  information
(borrowing  from the  language of  quantum theory,  when one
state from  a pairwise commuting ensemble  is made available
for any value attained by  the random variable), the optimal
strategy remains trivial: as noticed by Arikan~\cite{Ari96},
it consists of measuring the given state in the basis of the
common eigenvectors  of the ensemble, and  then querying the
values  in  the  non-increasing  order  of  their  posterior
probabilities.  The  triviality, in  this case,  steams from
the fact that, due to the linearity of the problem, there is
only  one   candidate  (i.e.    extremal)  as   the  optimal
measurement.

The guesswork problem has only very recently been tackled in
the   quantum  case~\cite{CCWF15,   HKDW22,  DBK22,   DBK23,
  AND23}. The operational setup we  consider here is that of
such references.  While previous works~\cite{CCWF15, HKDW22}
on  the  quantum  guesswork  continued along  the  lines  of
characterizing entropic properties of the guesswork, it must
be remarked that in the presence of quantum side information
the problem  of computing  the optimal guesswork  is instead
far  from  trivial  because   the  states  of  the  ensemble
generally do  not commute.  It has  been shown~\cite{HKDW22}
that   the   optimal   quantum  strategy   consists   of   a
single-measurement  that simultaneously  produces as  output
the optimal order in which  the queries should be performed;
technically, such  a strategy performs  as well as  the most
general sequential quantum-instrument  strategy; however, in
the  same reference  the optimization  of the  guesswork was
performed  only numerically,  and  hence approximately,  and
solely for a couple of qubit ensembles (trine and square).

Deriving  the optimal  quantum measurement~\cite{DBK22}  has
been  proven~\cite{DBK23}  to  be, in  general,  an  NP-hard
continuous optimization problem, and it is the focus of this
work.   In  particular,   in  Ref.~\cite{DBK22}   sufficient
conditions  were derived  for a  two-outcome measurement  to
attain  the  guesswork. It  must  be  noted that,  once  the
optimal  quantum  measurement has  been  found,  one has  an
updated posterior probability  distribution, and the problem
goes back to its classical version.

Formally, the  guesswork problem can be  conveniently framed
as a theoretical  game involving two parties,  say Alice and
Bob. At each round, Alice and  Bob are given a classical and
a quantum  state, respectively, the latter  solely dependant
on  the  former (this  naturally  corresponds  to a  quantum
ensemble or, equivalently,  a classical-quantum state).  Bob
queries Alice one state at a time until he correctly guesses
her  classical state,  at which  point he  pays a  cost that
solely depends on  the number of queries he  had to perform.
The  probability distribution  according to  which classical
states  are  sampled,  the  mapping  between  classical  and
quantum states, as  well as the cost function,  are known in
advance to the parties.

The  minimum  quantum  guesswork  problem  consists  of  the
minimization   of  Bob's   average  cost.    While  previous
works~\cite{HKDW22, DBK22}  were restricted  to the  case of
the identity  cost function,  here we consider  an arbitrary
balanced  cost function;  we remark  that, in  the classical
case, such an extension has been considered very recently by
Arslan and Haytaoglu~\cite{AH20}.   This scenario describes,
for instance, the case in  which Alice increases the cost of
successive  queries to  discourage  Bob from  a brute  force
attack.

When the  quantum state Bob  has access to corresponds  to a
two-level quantum  system, or  qubit, we recognize  that the
computation  of  the  minimum  guesswork,  by  definition  a
continuous optimization  problem over  quantum measurements,
is  equivalent  to  a  particular  instance  of  the  finite
combinatorial  optimization problem  known as  the quadratic
assignment  problem~\cite{KB57, SG76,  LM88, Cel98,  BCPP98,
  BCRW98} (QAP).  Specifically, given  that qubit states can
be represented  as vectors  in a ball  (the Bloch  ball), we
show that  the optimal measurement  is given by  finding the
maximum-length  linear combination  of  the  vectors of  the
ensemble with  weights given by  the values attained  by the
cost function.

Finally, we solve such an instance  of the QAP for the class
of so-called benevolent qubit  ensembles, that is, ensembles
whose Gram matrix is benevolent~\cite{Cel98}. It immediately
follows  from known  results on  the QAP  that any  ensemble
whose  image in  the  Bloch  ball is  a  regular polygon  is
benevolent (the guesswork of regular polygonal ensembles was
computed in  an independent way in  Ref.~\cite{DBK22}); here
we show that uniform anti-prismatic ensembles are benevolent
too.   Such a  class includes  ensembles that  commonly find
application   in  quantum   informational  tasks,   such  as
symmetric,  informationally  complete   (SIC)  and  mutually
unbiased  basis   (MUBs)  ensembles,  that   corresponds  to
tetrahedrons and octahedrons in the Bloch ball.

The hierarchical structure of our results is as follows:
\begin{itemize}
\item   Theorem~\ref{thm:helstrom}   solves  the   guesswork
  problem (that is, recasts the continuous optimization as a
  combinatorial,  and  therefore finite,  optimization)  for
  ensembles satisfying certain condition (also verifiable by
  finite     exhaustive    search)     spelled    out     by
  Eq.~\eqref{eq:condition}.
\item Corollary~\ref{cor:qubit}  answers on  the affirmative
  and  constructively  the  question whether  the  class  of
  ensembles      satisfying       the      condition      of
  Theorem~\ref{thm:helstrom} is  non-empty, by  showing that
  the  class of  qubit  ensembles  with uniform  probability
  distribution   satisfies    Eq.~\ref{eq:condition},   thus
  framing guesswork problem in  this case as a combinatorial
  problem known as quadratic assignment.
\item Corollary~\ref{cor:benevolent} goes  one step further,
  by  showing  that for  qubit  ensembles  that satisfy  the
  so-called   ``benevolence''   condition,   the   quadratic
  assignment (typically  an NP-hard  problem) can  be solved
  (e.g.,  the complexity  goes  from  $O(N!)$ to  constant).
\item Corollary~\ref{cor:antiprism}, once  again, answers on
  the  affirmative and  constructively the  question whether
  the class of ``benevolent''  qubit ensembles is non-empty,
  by  providing a  simple  geometrical  condition for  qubit
  ensembles to be benevolent, i.e. by showing that ensembles
  whose  image in  the  Bloch sphere  is  an anti-prism  are
  benevolent. This  includes SICs and MUBs,  qubit ensembles
  typically considered in quantum informational protocols.
\end{itemize}

\section{Formalization}

The notation adopted in this work is consistent with that of
Ref.~\cite{DBK22}.  We use  standard definitions and results
from  quantum  information  theory  (see e.g.   part  II  of
Ref.~\cite{Wil17}).  A quantum  state  is  represented by  a
density matrix  $\rho$, that  is, a unit-trace  (i.~e.  $\Tr
\rho  =  1$)  positive-semidefinite  (i.~e.  $\rho  \ge  0$)
operator on a Hilbert  space $\mathcal{H}$. A quantum effect
(the element of  a quantum measurement) is  represented by a
positive-semidefinite  ($\pi   \ge  0$)  operator   that  is
majorized by the identity ($\pi \le \openone$); the elements
of a  measurement sum up  to the identity.   The conditional
probability of outcome $\pi$ given  state $\rho$ is given by
the Born rule,  that is $P ( \pi  | \rho ) = \Tr  [ \rho \pi
]$.

For any  finite dimensional Hilbert space  $\mathcal{H}$, we
denote with $\mathcal{L} ( \mathcal{H} )$ and $\mathcal{L}_+
( \mathcal{H}  )$ the space  of Hermitian operators  and the
cone  of positive  semidefinite operators,  respectively, on
$\mathcal{H}$.

For any finite set  $\mathcal{M}$ and any finite-dimensional
Hilbert  space $\mathcal{H}$,  we denote  with $\mathcal{E}(
\mathcal{M},  \mathcal{H})$   the  set  of   ensembles  from
$\mathcal{M}$ to $\mathcal{H}$ given by
\begin{align*}
  \mathcal{E}  \left( \mathcal{M},  \mathcal{H} \right)  : =
  \left\{ \boldsymbol{\rho} :  \mathcal{M} \to \mathcal{L}_+
  \left( \mathcal{H} \right)  \Big| \sum_{m \in \mathcal{M}}
  \Tr \left[ \boldsymbol{\rho} \left(  m \right) \right] = 1
  \right\},
\end{align*}
that  is, ensembles  are defined  in terms  of subnormalized
states that sum up to a normalized state.

Let  us   first  introduce   the  most   general  sequential
instrument  strategy.  For  any  finite dimensional  Hilbert
spaces $\mathcal{H}_0$  and $\mathcal{H}_1$, we  denote with
$\mathcal{C}  \left(  \mathcal{H}_0, \mathcal{H}_1  \right)$
the   cone  of   completely   positive   linear  maps   from
$\mathcal{L}   (   \mathcal{H}_0   )$  to   $\mathcal{L}   (
\mathcal{H}_1 )$.

For any finite set  $\mathcal{M}$ and any finite-dimensional
Hilbert  spaces  $\mathcal{H}_0$   and  $\mathcal{H}_1$,  we
denote  with  $\mathcal{I}   (  \mathcal{M},  \mathcal{H}_0,
\mathcal{H}_1 )$  the set of instruments  from $\mathcal{M}$
to $\mathcal{C} ( \mathcal{H}_0, \mathcal{H}_1 )$ given by
\begin{align*}
  \mathcal{I}     \left(     \mathcal{M},     \mathcal{H}_0,
  \mathcal{H}_1  \right)  :=  \left\{  \boldsymbol{\epsilon}  :
  \mathcal{M}   \to   \mathcal{C}  \left(   \mathcal{H}_0,
  \mathcal{H}_1  \right)  \Big|   \sum_{m  \in  \mathcal{M}}
  \boldsymbol{\epsilon}   \left(  m   \right)\textrm{TP}
  \right\},
\end{align*}
where  TP   stands  for   trace-preserving  map,   that  is,
instruments are defined in terms of completely positive maps
that sum up to a trace preserving map.

For any finite set $\mathcal{M}$, any tuple $( \mathcal{H}_t
)_{t  = 0}^{|\mathcal{M}|}$  of  finite dimensional  Hilbert
spaces,  any ensemble  $\boldsymbol{\rho} \in  \mathcal{E} (
\mathcal{M},   \mathcal{H}_0   )$,    and   any   tuple   $(
\boldsymbol{\epsilon}_t    \in     \mathcal{I}    (\mathcal{M},
\mathcal{H}_{t-1}, \mathcal{H}_t) )_{t = 1}^{|\mathcal{M}|}$
of  instruments,  we  denote with  $p_{\boldsymbol{\rho},  (
  \boldsymbol{\epsilon}_t  )_t}$  the probability  distribution
that the outcomes of the sequential-instruments strategy are
$\mathbf{m}       =       (       \mathbf{m}_1,       \dots,
\mathbf{m}_{|\mathcal{M}|} )$  and that the $t$-th  query is
correct, that is
\begin{align*}
  &   p_{\boldsymbol{\rho},    \left(   \boldsymbol{\epsilon}_t
    \right)_t}  :   \mathcal{M}^{\times  \left|  \mathcal{M}
    \right|} \times \{ 1 , \dots, | \mathcal{M} | \} \to [0,
    1]    \\     &    \left(    \mathbf{m},     t    \right)
  \mapsto  \begin{cases} 0  &  \textrm{if $\mathbf{m}_t  \in
      \mathbf{m}_{\{1,  \dots,  t  -  1\}}$,}\\  \Tr  \left[
      \bigcirc_{s  =  1}^{t}  \boldsymbol{\epsilon}_{s}  \left(
      \mathbf{m}_{s}   \right)    \boldsymbol{\rho}   \left(
      \mathbf{m}_t \right) \right] & \textrm{otherwise,}
  \end{cases}
\end{align*}
where the symbol $\bigcirc$  denotes composition, or, in the
quantum circuit formalism,
\begin{align*}
  p_{\boldsymbol{\rho},      \left(     \boldsymbol{\epsilon}_t
    \right)_t} \left( \mathbf{m}, t \right) = \qquad
  \begin{aligned}
    \Qcircuit @C=.75em  @R=1.25em {  \lstick{\mathbf{m}_t} &
      \cw & \cw & \cghost{=} & \ustick{\mathbf{m}_t} \cw & &
      &  \cw  &  \cw &  \cghost{=}  &  \ustick{\mathbf{m}_t}
      \cw\\      &      \nghost{\boldsymbol{\epsilon}_1}      &
      \ustick{\mathbf{m}_1}   \cw  &   \cmultigate{-1}{=}  &
      \ustick{\textrm{no}}   \cw   &   \push{\cdots}   &   &
      \cghost{\boldsymbol{\epsilon}_t}  & \ustick{\mathbf{m}_t}
      \cw &  \cmultigate{-1}{=} &  \ustick{\textrm{yes}} \cw
      \\   \lstick{\boldsymbol{\rho}   \left(   \mathbf{m}_t
        \right)}  &  \multigate{-1}{\boldsymbol{\epsilon}_1}  &
      \qw      &     \qw      &      \qw      &     &      &
      \multigate{-1}{\boldsymbol{\epsilon}_t}  &  \qw &  \qw  &
      \qw}
  \end{aligned},
\end{align*}
for  any $\mathbf{m}  \in \mathcal{M}^{\times  | \mathcal{M}
  |}$ and  any $t \in \times  \{ 1 , \dots,  | \mathcal{M} |
\}$.    We    denote   with   $    q_{\boldsymbol{\rho},   (
  \boldsymbol{\epsilon}_t  )_t}$  the probability  distribution
that  the $t$-th  query is  correct, obtained  marginalizing
$p_{\boldsymbol{\rho}, ( \boldsymbol{\epsilon}_t )_t}$, that is
\begin{align*}
  &   q_{\boldsymbol{\rho},    \left(   \boldsymbol{\epsilon}_t
    \right)_t}  :  \left\{   1,  \dots,  \left|  \mathcal{M}
  \right| \right\} \to [0, 1]\\ & t \mapsto \sum_{\mathbf{m}
    \in  \mathcal{M}^{\times  \left|  \mathcal{M}  \right|}}
  p_{\boldsymbol{\rho},      \left(     \boldsymbol{\epsilon}_t
    \right)_t} \left( \mathbf{m}, t \right).
\end{align*}

For any finite set $\mathcal{M}$,  any function $\gamma : \{
1, \dots, |\mathcal{M}| \} \to \mathbb{R}$, and any tuple $(
\mathcal{H}_t   )_{t   =   0}^{|\mathcal{M}|}$   of   finite
dimensional     Hilbert    spaces,     we    denote     with
$G^{\gamma}_{\inst}$    the    guesswork    function    with
sequential-instruments  strategy, that  is, the  expectation
value of $\gamma$ given by
\begin{align*}
  &  G^{\gamma}_{\inst}  : \mathcal{E}  \left(  \mathcal{M},
  \mathcal{H}_0     \right)     \times    \bigtimes_{t     =
    1}^{|\mathcal{M}|}   \mathcal{I}   \left(   \mathcal{M},
  \mathcal{H}_{t-1},      \mathcal{H}_t     \right)      \to
  \mathbb{R}\\    &    \left(   \boldsymbol{\rho},    \left(
  \boldsymbol{\epsilon}_t \right)_t  \right) \mapsto  \sum_{t =
    1}^{\left|  \mathcal{M}   \right|}  \gamma\left(t\right)
  q_{\boldsymbol{\rho},      \left(     \boldsymbol{\epsilon}_t
    \right)_t} \left( t \right).
\end{align*}

\begin{dfn}[Minimum guesswork]
  For any  finite set $\mathcal{M}$, any  function $\gamma :
  \{  1, \dots,  |\mathcal{M}| \}  \to \mathbb{R}$,  and any
  finite  dimensional  Hilbert  space  $\mathcal{H}_0$,  the
  minimum guesswork $G_{\min}^{\gamma}$ is given by
  \begin{align*}
    &  G_{\min}^{\gamma} :  \mathcal{E} \left(  \mathcal{M},
    \mathcal{H} \right) \to \mathbb{R}\\ & \boldsymbol{\rho}
    \mapsto   \min_{\left(    \mathcal{H}_t   \right)_{t   =
        1}^{\left|           \mathcal{M}           \right|}}
    \min_{\left(\boldsymbol{\epsilon}  \in  \mathcal{I}  \left(
      \mathcal{M}, \mathcal{H}_{t-1},  \mathcal{H}_t \right)
      \right)_{t   =    1}^{\left|   \mathcal{M}   \right|}}
    G^{\gamma}_{\inst}    \left(     \boldsymbol{\rho},    (
    \boldsymbol{\epsilon}_t )_t \right).
  \end{align*}
\end{dfn}

Let  us   now  consider   the  case   of  single-measurement
strategy. For  any finite set $\mathcal{M}$,  we denote with
$\mathcal{N}_{\mathcal{M}}$   the  set   of  numberings   of
$\mathcal{M}$ given by
\begin{align*}
  \mathcal{N}_{\mathcal{M}} := \left\{  \mathbf{n} : \left\{
  1,   \dots,  \left|   \mathcal{M}  \right|   \right\}  \to
  \mathcal{M} \Big| \mathbf{n} \textrm{ bijective} \right\}.
\end{align*}
For  any  finite set  $\mathcal{M}$  and  any Hilbert  space
$\mathcal{H}$,     we      denote     with     $\mathcal{P}(
\mathcal{N}_{\mathcal{M}},   \mathcal{H}  )$   the  set   of
numbering-valued measurements given by
\begin{align*}
  \mathcal{P} \left(  \mathcal{N}_{\mathcal{M}}, \mathcal{H}
  \right)      :=       \left\{      \boldsymbol{\pi}      :
  \mathcal{N}_{\mathcal{M}}    \to   \mathcal{L}_+    \left(
  \mathcal{H}    \right)    \Big|    \sum_{\mathbf{n}    \in
    \mathcal{N}_{\mathcal{M}}}    \boldsymbol{\pi}    \left(
  \mathbf{n} \right) = \openone \right\}.
\end{align*}

For  any finite  set $\mathcal{M}$,  any finite  dimensional
Hilbert space $\mathcal{H}$, any ensemble $\boldsymbol{\rho}
\in   \mathcal{E}  (\mathcal{M},   \mathcal{H})$,  and   any
numbering-valued    measurement     $\boldsymbol{\pi}    \in
\mathcal{P}  ( \mathcal{N}_{\mathcal{M}},  \mathcal{H})$, we
denote  with  $p_{\boldsymbol{\rho}, \boldsymbol{\pi}}$  the
probability   distribution   that   the   outcome   of   the
single-measurement strategy  is $\mathbf{n}$ and  the $t$-th
query is correct, that is
\begin{align*}
  &      p_{\boldsymbol{\rho},      \boldsymbol{\pi}}      :
  \mathcal{N}_{\mathcal{M}}  \times  \left\{  1 ,  \dots,  |
  \mathcal{M} | \right\} \to [0, 1]\\ & \left( \mathbf{n}, t
  \right)   \mapsto  \Tr   \left[  \boldsymbol{\pi}   \left(
    \mathbf{n}  \right) \boldsymbol{\rho}  \left( \mathbf{n}
    \left( t \right) \right) \right],
\end{align*}
or, in the quantum circuit formalism,
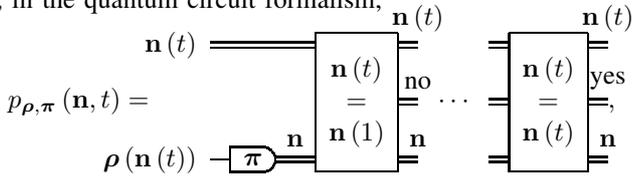
\begin{align*}
  p_{\boldsymbol{\rho}, \boldsymbol{\pi}} \left( \mathbf{n},
  t \right) = \qquad
  \begin{aligned}
    \Qcircuit @C=.75em @R=1.25em { \lstick{\mathbf{n} \left(
        t      \right)}     &      \cw      &     \cw      &
      \cghost{\begin{matrix}\mathbf{n}\left(t\right)\\=\\\mathbf{n}\left(1\right)\end{matrix}}&
      \ustick{\mathbf{n}  \left(  t  \right)}   \cw  &  &  &
      \cghost{\begin{matrix}\mathbf{n}\left(t\right)\\=\\\mathbf{n}\left(t\right)\end{matrix}}
      &  \ustick{\mathbf{n} \left(  t \right)}  \cw\\ &  & &
      \nghost{\begin{matrix}\mathbf{n}\left(t\right)\\=\\\mathbf{n}\left(1\right)\end{matrix}}
      &  \ustick{\textrm{no}}   \cw  &  \push{\cdots}   &  &
      \cghost{\begin{matrix}\mathbf{n}\left(t\right)\\=\\\mathbf{n}\left(t\right)\end{matrix}}
      &               \ustick{\textrm{yes}}              \cw
      \\ \lstick{\boldsymbol{\rho}  \left( \mathbf{n} \left(
        t \right)  \right)} &  \measureD{\boldsymbol{\pi}} &
      \ustick{\mathbf{n}}                \cw               &
      \cmultigate{-2}{\begin{matrix}\mathbf{n}\left(t\right)\\=\\\mathbf{n}\left(1\right)\end{matrix}}
      &      \ustick{\mathbf{n}}     \cw      &     &      &
      \cmultigate{-2}{\begin{matrix}\mathbf{n}\left(t\right)\\=\\\mathbf{n}\left(t\right)\end{matrix}}
      & \ustick{\mathbf{n}} \cw }
  \end{aligned},
\end{align*}
for any  $\mathbf{n} \in \mathcal{N}_{\mathcal{M}}$  and any
$t \in  \{ 1 ,  \dots, | \mathcal{M}  | \}$. We  denote with
$q_{\boldsymbol{\rho},  \boldsymbol{\pi}}$  the  probability
distribution  that the  $t$-th  guess  is correct,  obtained
marginalizing   $p_{\boldsymbol{\rho},   \boldsymbol{\pi}}$,
that is
\begin{align*}
  &  q_{\boldsymbol{\rho},  \boldsymbol{\pi}} :  \left\{  1,
  \dots, | \mathcal{M} | \right\} \to \left[ 0, 1\right]\\ &
  t \mapsto  \sum_{\mathbf{n} \in \mathcal{N}_{\mathcal{M}}}
  p_{\boldsymbol{\rho}, \boldsymbol{\pi}} \left( \mathbf{n},
  t \right),
\end{align*}

For any finite set $\mathcal{M}$,  any function $\gamma : \{
1, \dots, | \mathcal{M} | \} \to \mathbb{R}$, and any finite
dimensional  Hilbert  space  $\mathcal{H}$, we  denote  with
$G^\gamma$     the     \emph{guesswork}    function     with
single-measurement strategy, that  is, the expectation value
of $\gamma$ given by
\begin{align*}
   & G^{\gamma} : \mathcal{E} \left(\mathcal{M}, \mathcal{H}
  \right)         \times          \mathcal{P}         \left(
  \mathcal{N}_{\mathcal{M}},    \mathcal{H}   \right)    \to
  \mathbb{R}\\ &  \left( \boldsymbol{\rho}, \boldsymbol{\pi}
  \right) \mapsto \sum_{t  = 1}^{\left| \mathcal{M} \right|}
  \gamma    \left(     t    \right)    q_{\boldsymbol{\rho},
    \boldsymbol{\pi}} \left(t\right).
\end{align*}

The  following lemma,  already proved  in Ref.~\cite{HKDW22}
(see  Theorem~1  and  Eq.~(53)   therein),  shows  that  the
single-measurement  strategy  is   equivalent  to  the  most
general sequential instrument strategy.

\begin{lmm}
  For any  finite set $\mathcal{M}$, any  function $\gamma :
  \{ 1, \dots, |\mathcal{M}|  \} \to \mathbb{R}$, any finite
  dimensional   Hilbert  space   $\mathcal{H}_0$,  and   any
  ensemble $\boldsymbol{\rho}  \in \mathcal{E} (\mathcal{M},
  \mathcal{H})$,  the minimum  guesswork $G_{\min}^{\gamma}$
  satisfies
  \begin{align*}
    G_{\min}^{\gamma}  \left(  \boldsymbol{\rho}  \right)  =
    \min_{\boldsymbol{\pi}     \in    \mathcal{P}     \left(
      \mathcal{N}_{\mathcal{M}},    \mathcal{H}_0   \right)}
    G^{\gamma}  \left(  \boldsymbol{\rho},  \boldsymbol{\pi}
    \right).
  \end{align*}
\end{lmm}

\begin{proof}
  Notice first  that without loss of  generality the optimal
  $(  \boldsymbol{\epsilon}^*_t  )_t$  is such  that  the  same
  outcome is produced only once, that is
  \begin{align*}
    \Tr   \left[  \boldsymbol{\epsilon}^*_{\left|   \mathcal{M}
        \right|}   \left(   \mathbf{m}_{\left|   \mathcal{M}
        \right|} \right) \dots \boldsymbol{\epsilon}^*_1 \left(
      \mathbf{m}_1    \right)    \boldsymbol{\rho}    \left(
      \boldsymbol{m} \right) \right] = 0
  \end{align*}
  if  there   exists  $t_1   \neq  t_2   \in  \{   1,  \dots
  |\mathcal{M}|   \}$   such    that   $\mathbf{m}_{t_1}   =
  \mathbf{m}_{t_2}$.     Clearly,   for    any   tuple    $(
  \boldsymbol{\epsilon}_t  )_t$  of instruments,  there  exists
  numbering-valued measurement $\boldsymbol{\pi}$ such that
  \begin{align}
    \label{eq:equivalence}
    p_{\boldsymbol{\rho},  \left(  \boldsymbol{\epsilon}_t  \right)_t}
    \Big|_{\mathcal{N}_{\mathcal{M}} \times \left\{ 1 , \dots, |
      \mathcal{M}   |   \right\}}  =   p_{\boldsymbol{\rho},
      \boldsymbol{\pi}}.
  \end{align}
  Since   without  loss   of  generality   the  optimal   $(
  \boldsymbol{\epsilon}^*_t  )_t$ does  not  depend on  Alice's
  feedback, also  the reverse  implication is  true, namely,
  for  any  numbering-valued measurement  $\boldsymbol{\pi}$
  there  exists   tuple  $(  \boldsymbol{\epsilon}_t   )_t$  of
  instruments   such   that  Eq.~\eqref{eq:equivalence}   is
  verified.
\end{proof}

\section{Main result}

\subsection{Arbitrary dimensional case}

The following lemma shows that the minimum guesswork problem
is covariant under  permutations of the ensemble  and of the
cost function.

\begin{lmm}
  \label{lmm:covariance}
  For any  finite set $\mathcal{M}$, any  function $\gamma :
  \{ 1, \dots, |\mathcal{M}|  \} \to \mathbb{R}$, any finite
  dimensional Hilbert space  $\mathcal{H}$, and any ensemble
  $\boldsymbol{\rho}    \in     \mathcal{E}    (\mathcal{M},
  \mathcal{H})$,   if   measurement  $\boldsymbol{\pi}   \in
  \mathcal{P}  (   \mathcal{N}_{\mathcal{M}},  \mathcal{H})$
  attains   the  minimum   guesswork  $G^{\gamma}_{\min}   (
  \boldsymbol{\rho}   )$,  that   is  $G^{\gamma}_{\min}   (
  \boldsymbol{\rho}  )  =  G^{\gamma}  (  \boldsymbol{\rho},
  \boldsymbol{\pi}  )$, then  measurement $\boldsymbol{\pi}'
  \in \mathcal{P} ( \mathcal{N}_{\mathcal{M}}, \mathcal{H})$
  attains the  minimum guesswork  for cost  function $\gamma
  \circ  \sigma_1$  and  ensemble  $\boldsymbol{\rho}  \circ
  \sigma_2$,  that is  $G^{\gamma  \circ \sigma_1}_{\min}  (
  \boldsymbol{\rho}  \circ   \sigma_2)  =   G^{\gamma  \circ
    \sigma_1}    (    \boldsymbol{\rho}   \circ    \sigma_2,
  \boldsymbol{\pi}'    )$,   where    $\boldsymbol{\pi}'   (
  \mathbf{n}   )  :=   \boldsymbol{\pi}  (   \sigma_2  \circ
  \mathbf{n} \circ \sigma_1^{-1} )$  for any $\mathbf{n} \in
  \mathcal{N}_{\mathcal{M}}$,    for     any    permutations
  $\sigma_1$ and $\sigma_2$.
\end{lmm}

\begin{proof}
  One has
  \begin{align*}
    &  G^{\gamma  \circ \sigma_1}  \left(  \boldsymbol{\rho}
    \circ \sigma_2, \boldsymbol{\pi} \right)  \\ = & \sum_{t
      = 1}^{\left| \mathcal{M} \right|} \sum_{\mathbf{n} \in
      \mathcal{N}_{\mathcal{M}}}   \gamma   \circ   \sigma_1
    \left(  t  \right)  \Tr \left[  \boldsymbol{\pi}  \left(
      \mathbf{n}  \right)  \boldsymbol{\rho} \circ  \sigma_2
      \left(  \mathbf{n} \left(  t  \right) \right)  \right]
    \\  =  &  \sum_{t'  =  1}^{\left|  \mathcal{M}  \right|}
    \sum_{\mathbf{n}  \in \mathcal{N}_{\mathcal{M}}}  \gamma
    \left(  t' \right)  \Tr  \left[ \boldsymbol{\pi}  \left(
      \mathbf{n}  \right) \boldsymbol{\rho}  \left( \sigma_2
      \circ \mathbf{n} \circ \sigma_1^{-1} \left( t' \right)
      \right) \right]
  \end{align*}
  where the  second equality follows by  the replacement $t'
  :=  \sigma_1  (  t  )$. Hence  the  statement  immediately
  follows.
\end{proof}

For any discrete function $f$, we denote with $\overline{f}$
its average. For any  finite set $\mathcal{M}$, any function
$\gamma : \{  1, \dots, | \mathcal{M} |  \} \to \mathbb{R}$,
any finite-dimensional Hilbert  space $\mathcal{H}$, and any
ensemble  $\boldsymbol{\rho}   \in  \mathcal{E}(\mathcal{M},
\mathcal{H})$,           we            denote           with
$E^{\gamma}_{\boldsymbol{\rho}}  : \mathcal{N}_{\mathcal{M}}
\to \mathcal{L} ( \mathcal{H} )$ the map given by
\begin{align}
  \label{eq:effective}
  E^{\gamma}_{\boldsymbol{\rho}} \left( \mathbf{n} \right) &
  :=  2 \sum_{t  =  1}^{\left|  \mathcal{M} \right|}  \left(
  \gamma   \left(   t    \right)   -   \overline{   \gamma
  } \right)  \boldsymbol{\rho} \left( \mathbf{n}
  \left( t \right) \right),
\end{align}
for any $\mathbf{n}  \in \mathcal{N}_{\mathcal{M}}$.

The   following  lemma   expresses   the   guesswork  as   a
discrimination           problem           for           the
$E^{\gamma}_{\boldsymbol{\rho}}$'s operators, thus providing
an interpretation for such operators.

\begin{lmm}
  \label{lmm:gw}
  For any  finite set $\mathcal{M}$, any  function $\gamma :
  \{ 1, \dots, |\mathcal{M}|  \} \to \mathbb{R}$, any finite
  dimensional  Hilbert  space  $\mathcal{H}$,  any  ensemble
  $\boldsymbol{\rho}    \in     \mathcal{E}    (\mathcal{M},
  \mathcal{H})$,   and   any  numbering-valued   measurement
  $\boldsymbol{\pi}         \in        \mathcal{P}         (
  \mathcal{N}_{\mathcal{M}},  \mathcal{H})$,  the  guesswork
  $G^{\gamma}  (  \boldsymbol{\rho},  \boldsymbol{\pi})$  is
  given by
  \begin{align*}
    G^{\gamma}  \left(  \boldsymbol{\rho},  \boldsymbol{\pi}
    \right) = \overline{ \gamma } + \frac12 \sum_{\mathbf{n}
      \in     \mathcal{N}_{\mathcal{M}}}      \Tr     \left[
      \boldsymbol{\pi}     \left(     \mathbf{n}     \right)
      E^{\gamma}_{\boldsymbol{\rho}}    \left(    \mathbf{n}
      \right) \right].
  \end{align*}
\end{lmm}

\begin{proof}
  By the  linearity of  map $E^{\gamma}_{\boldsymbol{\rho}}$
  one has
  \begin{align*}
    & G^{\gamma}  \left( \boldsymbol{\rho}, \boldsymbol{\pi}
    \right) \\ = & \sum_{t = 1}^{\left| \mathcal{M} \right|}
    \sum_{\mathbf{n}  \in \mathcal{N}_{\mathcal{M}}}  \gamma
    \left(  t  \right)  \Tr \left[  \boldsymbol{\pi}  \left(
      \mathbf{n} \right) \boldsymbol{\rho} \left( \mathbf{n}
      \left(    t   \right)    \right)    \right]\\   =    &
    \sum_{\mathbf{n}   \in  \mathcal{N}_{\mathcal{M}}}   \Tr
    \left[   \boldsymbol{\pi}   \left(  \mathbf{n}   \right)
      \sum_{t  =  1}^{\left|   \mathcal{M}  \right|}  \left(
      \gamma  \left(  t   \right)  \boldsymbol{\rho}  \left(
      \mathbf{n}   \left(   t    \right)   \right)   \right)
      \right]\\      =      &      \sum_{\mathbf{n}      \in
      \mathcal{N}_{\mathcal{M}}} \Tr \left[ \boldsymbol{\pi}
      \left(        \mathbf{n}         \right)        \left(
      \frac{E^{\gamma}_{\boldsymbol{\rho}} \left( \mathbf{n}
        \right)}2  +  \overline{   \gamma  }  \;  \overline{
        \boldsymbol{\rho} } \right) \right],
  \end{align*}
  from which the statement immediately follows.
\end{proof}

We say  that function  $\gamma$ is balanced  if and  only if
there exists permutation $\sigma_{\gamma}$ such that
\begin{align}
  \label{eq:balanced}
  \frac{\gamma  +  \gamma  \circ \sigma_{\gamma}^{-1}  }2  =
  \overline{ \gamma }.
\end{align}
The identity cost function is recovered for $\sigma_{\gamma}
:= ( |\mathcal{M}|, \dots, 1)$.

\begin{lmm}
  \label{lmm:balanced}
  For any  finite set  $\mathcal{M}$, any  balanced function
  $\gamma :  \{ 1, \dots, |\mathcal{M}|  \} \to \mathbb{R}$,
  any  finite dimensional  Hilbert space  $\mathcal{H}$, and
  any    ensemble    $\boldsymbol{\rho}   \in    \mathcal{E}
  (\mathcal{M},  \mathcal{H})$, one has
  \begin{align*}
    E^{\gamma}_{\boldsymbol{\rho}} \left( \mathbf{n} \right) =
    -E^{\gamma}_{\boldsymbol{\rho}}  \left(  \mathbf{n}  \circ
    \sigma_{\gamma} \right),
  \end{align*}
  for any $\mathbf{n} \in \mathcal{N}_{\mathcal{M}}$.
\end{lmm}

\begin{proof}
  One has
  \begin{align*}
    &   E_{\boldsymbol{\rho}}^{\gamma}   \left(   \mathbf{n}
    \right)  \\  =  &  2 \sum_{t  =  1}^{\left|  \mathcal{M}
      \right|}  \left( \overline{  \gamma }  - \gamma  \circ
    \sigma_{\gamma}^{-1}    \left(    t   \right)    \right)
    \boldsymbol{\rho}  \left(  \mathbf{n} \left(  t  \right)
    \right)\\  =  &  2  \sum_{t'  =  1}^{\left|  \mathcal{M}
      \right|} \left( \overline{ \gamma } - \gamma \left( t'
    \right)  \right)   \boldsymbol{\rho}  \left(  \mathbf{n}
    \circ \sigma_{\gamma} \left( t'  \right) \right)\\ = & -
    E_{\boldsymbol{\rho}}^{\gamma}  \left( \mathbf{n}  \circ
    \sigma_{\gamma} \right),
  \end{align*}
  where     the     first      equality     follows     from
  Eq.~\eqref{eq:balanced},  the second  equality follows  by
  substituting  $t' :=  \sigma_{\gamma}^{-1} (  t )$,  and the  third
  equality follows from Eq.~\eqref{eq:effective}.
\end{proof}

We denote with  $\Pi_- ( \cdot )$ and $\Pi_0  ( \cdot )$ the
projectors on  the negative and  null parts of $(  \cdot )$,
respectively. For any finite set $\mathcal{M}$, any balanced
function  $\gamma :  \{ 1,  \dots,  | \mathcal{M}  | \}  \to
\mathbb{R}$,    any    finite-dimensional   Hilbert    space
$\mathcal{H}$,    any   ensemble    $\boldsymbol{\rho}   \in
\mathcal{E}(\mathcal{M},     \mathcal{H})$,      and     any
$\mathbf{n}^*  \in  \mathcal{N}_{ \mathcal{M}}$,  we  denote
with          $\boldsymbol{\pi}^{\gamma}_{\boldsymbol{\rho},
  \mathbf{n}^*}$ the numbering-valued measurement given by
\begin{align*}
  \boldsymbol{\pi}^{\gamma}_{\boldsymbol{\rho},
    \mathbf{n}^*} \left( \mathbf{n} \right) :=
  \begin{cases}
    \left(     \Pi_-     +     \frac12     \Pi_0     \right)
    E^{\gamma}_{\boldsymbol{\rho}} \left( \mathbf{n} \right)
    &   \textrm{if  }   \mathbf{n}   \in  \{   \mathbf{n}^*,
    \mathbf{n}^*   \circ    \sigma_{\gamma}   \},\\    0   &
    \textrm{otherwise},
  \end{cases}
\end{align*}
for any $\mathbf{n}  \in \mathcal{N}_{ \mathcal{M}}$ (notice
that   this   is   a   well-defined   measurement   due   to
Lemma~\eqref{lmm:balanced}).

The following theorem provides the minimum guesswork under a
finite  set of  conditions, thus  generalizing Theorem~1  of
Ref.~\cite{DBK22} to the case of balanced cost functions.

\begin{thm}
  \label{thm:helstrom}
  For any  finite set  $\mathcal{M}$, any  balanced function
  $\gamma : \{ 1, \dots, | \mathcal{M} | \} \to \mathbb{R}$,
  any  finite dimensional  Hilbert space  $\mathcal{H}$, and
  any   ensemble   $\boldsymbol{\rho}  \in   \mathcal{E}   (
  \mathcal{M},  \mathcal{H} )$,  if  there exists  numbering
  $\mathbf{n}^* \in \mathcal{N} ( \mathcal{M} )$ such that
  \begin{align}
    \label{eq:condition}
    \left|       E^{\gamma}_{\boldsymbol{\rho}}       \left(
    \mathbf{n}^*         \right)         \right|         \ge
    E^{\gamma}_{\boldsymbol{\rho}}     \left(     \mathbf{n}
    \right),
  \end{align}
  for any  $\mathbf{n} \in  \mathcal{N}_{\mathcal{M}}$, then
  numbering-valued                               measurement
  $\boldsymbol{\pi}_{\boldsymbol{\rho},   \mathbf{n}^*}  \in
  \mathcal{P}  (  \mathcal{N}_{\mathcal{M}}, \mathcal{H}  )$
  minimizes  the  guesswork,  that is  $G^{\gamma}_{\min}  (
  \boldsymbol{\rho}  )  =  G^{\gamma}  (  \boldsymbol{\rho},
  \boldsymbol{\pi}_{\boldsymbol{\rho},   \mathbf{n}^*}   )$,
  with
  \begin{align}
    \label{eq:gwmin}
    G^{\gamma}           \left(           \boldsymbol{\rho},
    \boldsymbol{\pi}_{\boldsymbol{\rho},       \mathbf{n}^*}
    \right)   =   \overline{\gamma}    -   \frac12   \left\|
    E^{\gamma}_{\boldsymbol{\rho}}    \left(    \mathbf{n}^*
    \right) \right\|_1.
  \end{align}
\end{thm}

\begin{proof}
  Equation~\eqref{eq:gwmin}  immediately  follows by  direct
  computation using Lemma~\ref{lmm:gw}. One has
  \begin{align*}
    &  G^{\gamma}_{\min}  \left(  \boldsymbol{\rho}  \right)
    \\     =      &     \overline{\gamma}      +     \frac12
    \min_{\boldsymbol{\pi}     \in    \mathcal{P}     \left(
      \mathcal{N}_{\mathcal{M}},     \mathcal{H}    \right)}
    \sum_{\mathbf{n}   \in  \mathcal{N}_{\mathcal{M}}}   \Tr
    \left[ E^{\gamma}_{\boldsymbol{\rho}}  \left( \mathbf{n}
      \right)  \boldsymbol{\pi}  \left(  \mathbf{n}  \right)
      \right]   \\   =   &   \overline{\gamma}   -   \frac12
    \max_{\boldsymbol{\pi}     \in    \mathcal{P}     \left(
      \mathcal{N}_{\mathcal{M}},     \mathcal{H}    \right)}
    \sum_{\mathbf{n}   \in  \mathcal{N}_{\mathcal{M}}}   \Tr
    \left[ E^{\gamma}_{\boldsymbol{\rho}}  \left( \mathbf{n}
      \right)  \boldsymbol{\pi}  \left(  \mathbf{n}  \right)
      \right]   \\  \ge   &   \overline{\gamma}  -   \frac12
    \max_{\boldsymbol{\pi}     \in    \mathcal{P}     \left(
      \mathcal{N}_{\mathcal{M}},     \mathcal{H}    \right)}
    \sum_{\mathbf{n}   \in  \mathcal{N}_{\mathcal{M}}}   \Tr
    \left[   \left|  E^{\gamma}_{\boldsymbol{\rho}}   \left(
      \mathbf{n}^*  \right) \right|  \boldsymbol{\pi} \left(
      \mathbf{n} \right) \right] \\  = & \overline{\gamma} -
    \frac12  \left\|  E^{\gamma}_{\boldsymbol{\rho}}  \left(
    \mathbf{n}^* \right) \right\|_1,
  \end{align*}
  where the first  equality follows from Lemma~\ref{lmm:gw},
  the  second  equality  follows  from  the  fact  that,  by
  Lemma~\ref{lmm:balanced},    one    has    $\sum_{n    \in
    \mathcal{N}_{\mathcal{M}}}             \Tr             [
    E_{\boldsymbol{\rho}}^{\gamma}     (    \mathbf{n}     )
    \boldsymbol{\pi}  (  \mathbf{n}  )]   =  -  \sum_{n  \in
    \mathcal{N}_{\mathcal{M}}}             \Tr             [
    E_{\boldsymbol{\rho}}^{\gamma}     (    \mathbf{n}     )
    \boldsymbol{\pi}        (        \mathbf{n}        \circ
    \sigma_{\gamma}^{-1})]$,  the  inequality  follows  from
  Eq.~\eqref{eq:condition}   and   the   non-negativity   of
  $\boldsymbol{\pi} ( \mathbf{n} )$  for any $\mathbf{n} \in
  \mathcal{N}_{\mathcal{M}}$, and the final equality follows
  by direct  computation. Hence, the statement  follows from
  Eq.~\eqref{eq:gwmin}.
\end{proof}

Therefore, given  the task of  computing the guesswork  of a
given  ensemble  $\boldsymbol{\rho}$,   one  can  check  the
conditions in Eq.~\eqref{eq:condition} that, being finite in
number, can be checked in  finite time. If they are verified
for  some  numbering  $\mathbf{n}^*$,  the  numbering-valued
measurement            $\boldsymbol{\pi}_{\boldsymbol{\rho},
  \mathbf{n}^*}$   corresponds   to  the   optimal   quantum
strategy,  and  the  value  of the  guesswork  is  given  by
Eq.~\eqref{eq:gwmin}. An  even more concrete  application of
Theorem~\ref{thm:helstrom} is given in  the next section for
the qubit case.

\subsection{The qubit case}

For any function $f : \{ 1, \dots, M \} \to K$, where $K$ is
a linear space,  we denote with $\gram (f)$  the Gram matrix
whose  element $i,  j$ is  $f(i) \cdot  f(j)$.  For  any set
$\mathcal{M}  =  \{  1,  \dots, |\mathcal{M}|  \}$  and  any
numbering  $\mathbf{n}  \in  \mathcal{N}_{\mathcal{M}}$,  we
denote  with  $X_{\mathbf{n}}$   the  $|\mathcal{M}|  \times
|\mathcal{M}|$ permutation  matrix whose  element $i,  j$ is
$\delta_{i,      \mathbf{n}     (j)}$.       Notice     that
$X_{\mathcal{N}_{\mathcal{M}}}$   is   the    set   of   all
$|\mathcal{M}|     \times     |\mathcal{M}|$     permutation
matrices.    For   any    two-dimensional   Hilbert    space
$\mathcal{H}$, let $\mathbf{v}$ be the Pauli vector function
given by
\begin{align*}
  & \mathbf{v} : \mathcal{L}  \left( \mathcal{H} \right) \to
  \mathbb{R}^3\\ &  A \mapsto  \left( \Tr \left[  A \sigma_k
    \right] \right)_{k = 1}^3\nonumber,
\end{align*}
for  some orthonormal  basis $(  \sigma_k )_{k=1}^3$  in the
traceless subspace of $\mathcal{L} ( \mathcal{H} )$ equipped
with the Hilbert-Schmidt product.

In  the  qubit  case,  the following  corollary  provides  a
closed-form solution to the  guesswork problem (a continuous
optimization problem)  by proving  its equivalence  with the
quadratic assignment problem (an  optimization over a finite
set).

\begin{cor}
  \label{cor:qubit}
  For any  set $\mathcal{M}  = \{  1, \dots,  |\mathcal{M} |
  \}$,  any  balanced function  $\gamma  :  \{ 1,  \dots,  |
  \mathcal{M}  | \}  \to  \mathbb{R}$,  any two  dimensional
  Hilbert    space   $\mathcal{H}$,    and   any    ensemble
  $\boldsymbol{\rho}   \in    \mathcal{E}   (   \mathcal{M},
  \mathcal{H}   )$   such   that   the   prior   probability
  distribution   $\Tr[  \boldsymbol{\rho}   (  \cdot   )]  =
  |\mathcal{M}|^{-1}$    is    uniform,   the    measurement
  $\boldsymbol{\pi}^{\gamma}_{\boldsymbol{\rho},
    \mathbf{n}^*}$    attains    the    minimum    guesswork
  $G^{\gamma}_{\min}   (  \boldsymbol{\rho}   )$,  that   is
  $G^{\gamma}_{\min}  ( \boldsymbol{\rho}  ) =  G^{\gamma} (
  \boldsymbol{\rho},
  \boldsymbol{\pi}^{\gamma}_{\mathbf{n}^*}     )$,     where
  $\mathbf{n}^*$ is the solution  of the following quadratic
  assignment problem
  \begin{align*}
    \mathbf{n}^*       =       \argmax_{\mathbf{n}       \in
      \mathcal{N}_{\mathcal{M}}}   \Tr  \left[\gram   \left(
      \gamma  -   \overline{\gamma}  \right)  X_{\mathbf{n}}
      \gram   \left(   \mathbf{v}  \circ   \boldsymbol{\rho}
      \right) X_{\mathbf{n}}^T \right],
  \end{align*}
  and one has
  \begin{align*}
    G^{\gamma}           \left(           \boldsymbol{\rho},
    \boldsymbol{\pi}_{\boldsymbol{\rho},       \mathbf{n}^*}
    \right) = \overline{\gamma}  - \frac12 \left| \mathbf{v}
    \left(       E^{\gamma}_{\boldsymbol{\rho}}       \left(
    \mathbf{n}^* \right) \right) \right|_2,
  \end{align*}
  where $| \cdot |$ denotes the $2$-norm.
\end{cor}

\begin{proof}
  Since by hypothesis $\Tr [ \boldsymbol{\rho} ( \cdot ) ] =
  | \mathcal{M} |^{-1}$, by explicit computation one has
  \begin{align*}
    \Tr    \left[   E^{\gamma}_{\boldsymbol{\rho}}    \left(
      \mathbf{n} \right) \right] = 0,
  \end{align*}
  for   any   $\mathbf{n}  \in   \mathcal{N}_{\mathcal{M}}$.
  Hence,    since    by    hypothesis    $\mathcal{H}$    is
  two-dimensional, one has
  \begin{align*}
    \left| E^{\gamma}_{\boldsymbol{\rho}}  \left( \mathbf{n}
    \right) \right| = \left\| E^{\gamma}_{\boldsymbol{\rho}}
    \left( \mathbf{n} \right) \right\|_1 \frac{\openone}2,
  \end{align*}
  for             any            $\mathbf{n}             \in
  \mathcal{N}_{\mathcal{M}}$.  Therefore, the  range $\left|
  E^{\gamma}_{\boldsymbol{\rho}}  \left( \mathcal{N}  \left(
  \mathcal{M} \right)  \right) \right|$ is  totally ordered.
  Hence, there exists $\mathbf{n}^*$ such that
  \begin{align*}
    \left|       E^{\gamma}_{\boldsymbol{\rho}}       \left(
    \mathbf{n}^*     \right)      \right|     \ge     \left|
    E^{\gamma}_{\boldsymbol{\rho}} \left( \mathbf{n} \right)
    \right|   \ge    E^{\gamma}_{\boldsymbol{\rho}}   \left(
    \mathbf{n} \right),
  \end{align*}
  for   any   $\mathbf{n}  \in   \mathcal{N}_{\mathcal{M}}$.
  Hence,  due  to  Theorem~\ref{thm:helstrom},  the  minimum
  guesswork  $G^{\gamma}_{\min}  ( \boldsymbol{\rho}  )$  is
  given by Eq.~\eqref{eq:gwmin}.  The statement then follows
  by observing that by explicit computation
  \begin{align*}
    &    \left\|    E^{\gamma}_{\boldsymbol{\rho}}    \left(
    \mathbf{n}^*  \right) \right\|_1^2  = \left|  \mathbf{v}
    \left(       E^{\gamma}_{\boldsymbol{\rho}}       \left(
    \mathbf{n}^*  \right) \right)  \right|_2^2  \\  = & 4 \Tr
    \left[\gram  \left( \gamma  - \overline{\gamma}  \right)
      X_{\mathbf{n}}    \gram   \left(    \mathbf{v}   \circ
      \boldsymbol{\rho} \right) X_{\mathbf{n}}^T \right].
  \end{align*}
\end{proof}

In general, the solution of the quadratic assignment problem
in  Corollary~\ref{cor:qubit} requires  an exhaustive  search
over  $\mathcal{N}_{\mathcal{M}}$ (whose  cardinality is  $|
\mathcal{M}|!$), and can be carried  out by using one of the
several  known  algorithms~\cite{Cel98}  for  the  quadratic
assignment   problem.   In   the  following,   however,   we
investigate  situations  in which  such  a  solution can  be
singled-out a priori.

A square  matrix $A$ is  Toeplitz if  and only if  its entry
$A_{i, j}$ only  depends on $i -  j$ for any $i$  and $j$. A
matrix $A$ is benevolent~\cite{BCRW98} if  and only if it is
symmetric, Toeplitz, and satisfies the following properties:
\begin{enumerate}[Property 1]
  \item\label{item:monotone}   $A_{m   +   1,   1}$   is   a
    non-decreasing function of $m$ in $\{1, \dots, \floor{ |
      \mathcal{M} | / 2} \}$.
  \item\label{item:majorize} $A_{| \mathcal{M} | + 1 - m, 1}
    \ge A_{m + 1, 1}$ for  any $m \in \{1, \dots, \floor{( |
      \mathcal{M} | / 2} \}$.
\end{enumerate}

For any set $\mathcal{M} = \{ 1, \dots, |\mathcal{M}| \}$ we
define the numbering $\mathbf{n}_{\bv}$ given by
\begin{align*}
  \left(  \mathbf{n}_{\bv}  \right)^{-1}  \left(  m  \right)
  := \begin{cases}  2 m -  1 & \textrm{  if $m \le  \ceil{ |
        \mathcal{M} |  /2}$},\\ 2 \left(  \left| \mathcal{M}
    \right| + 1 - m \right) & \textrm{ otherwise}.
    \end{cases}
\end{align*}

The following  corollary leverages  known results  about the
quadratic   assignment   problem   to   solve   the   finite
optimization in  Corollary~\ref{cor:qubit} for the  class of
benevolent ensembles.

\begin{cor}
  \label{cor:benevolent}
  For any set $\mathcal{M} = \{ 1, \dots, |\mathcal{M}| \}$,
  any balanced function $\gamma : \{ 1, \dots, |\mathcal{M}|
  \}  \to  \mathbb{R}$,  any two-dimensional  Hilbert  space
  $\mathcal{H}$,  and  any ensemble  $\boldsymbol{\rho}  \in
  \mathcal{E}  (\mathcal{M},  \mathcal{H})$  such  that  the
  prior probability  distribution $\Tr[  \boldsymbol{\rho} (
    \cdot )] = |\mathcal{M}|^{-1}$ is uniform, if $- \gram (
  \mathbf{v} \circ  \boldsymbol{\rho} )$  is permutationally
  equivalent   to  a   benevolent  matrix   and  $\mathbf{v}
  (\overline{\boldsymbol{\rho}})    \cdot    \mathbf{v}    (
  \boldsymbol{\rho}    (\cdot))$     is    constant,    with
  $\overline{\boldsymbol{\rho}} :=  \sum_{m \in \mathcal{M}}
  \boldsymbol{\rho}    (     m    )$,     the    measurement
  $\boldsymbol{\pi}^{\gamma}_{\boldsymbol{\rho},
    \mathbf{n}^*   }$   attains    the   minimum   guesswork
  $G^{\gamma}_{\min}   (  \boldsymbol{\rho}   )$,  that   is
  $G^{\gamma}_{\min}  ( \boldsymbol{\rho}  ) =  G^{\gamma} (
  \boldsymbol{\rho},
  \boldsymbol{\pi}^{\gamma}_{\boldsymbol{\rho},
    \mathbf{n}^*} )$,  where $\mathbf{n}^* =  \sigma_2 \circ
  \mathbf{n}_{\bv}  \circ  \sigma_1^{-1}$  and  permutations
  $\sigma_1, \sigma_2$ are such that $\gamma \circ \sigma_1$
  is   non-decreasing  and   $-\gram   (  \mathbf{v}   \circ
  \boldsymbol{\rho} \circ \sigma_2 )$ is benevolent.
\end{cor}

\begin{proof}
  Let  us  consider first  the  case  in which  $\gamma$  is
  non-decreasing   and   $-   \gram   (   \mathbf{v}   \circ
  \boldsymbol{\rho})$  is   benevolent.   For  any   $k  \in
  \mathbb{R}$, by explicit computation one has
  \begin{align*}
    E_{\boldsymbol{\rho}}^{\gamma}  \left( \mathbf{n}\right)
    =  2 \sum_{t  = 1}^{\left|  \mathcal{M} \right|}  \left(
    \gamma \left( t \right)  - \overline{\gamma} + k \right)
    \boldsymbol{\rho}  \left(  \mathbf{n} \left(  t  \right)
    \right) - 2 k \overline{\boldsymbol{\rho}},
  \end{align*}
  for any  $\mathbf{n} \in  \mathcal{N}_{\mathcal{M}}$, from
  which it immediately follows that
  \begin{align*}
    &           \left|           \mathbf{v}           \left(
    E_{\boldsymbol{\rho}}^{\gamma}  \left( \mathbf{n}\right)
    \right) \right|_2^2  \\ =  & 4 \left|  \mathbf{v} \left(
    \sum_{t = 1}^{\left|  \mathcal{M} \right|} \left( \gamma
    \left(  t  \right)  -   \overline{\gamma}  +  k  \right)
    \boldsymbol{\rho}  \left(  \mathbf{n} \left(  t  \right)
    \right)     \right)     -    k     \mathbf{v}     \left(
    \overline{\boldsymbol{\rho}} \right) \right|_2^2,
  \end{align*}
  for any $\mathbf{n} \in \mathcal{N}_{\mathcal{M}}$.  Since
  $\mathbf{v}      (\overline{\boldsymbol{\rho}})      \cdot
  \mathbf{v} (  \boldsymbol{\rho} (\cdot))$ is  constant one
  has
  \begin{align*}
    &  \argmax_{\mathbf{n}   \in  \mathcal{N}_{\mathcal{M}}}
    \left| \mathbf{v}  \left( E_{\boldsymbol{\rho}}^{\gamma}
    \left(  \mathbf{n}\right)  \right)   \right|_2^2\\  =  &
    \argmax_{\mathbf{n}    \in    \mathcal{N}_{\mathcal{M}}}
    \left|   \mathbf{v}   \left(    \sum_{t   =   1}^{\left|
      \mathcal{M} \right|} \left( \gamma  \left( t \right) -
    \overline{\gamma} +  k \right)  \boldsymbol{\rho} \left(
    \mathbf{n} \left( t \right) \right) \right) \right|_2^2,
  \end{align*}
  for  any  $\mathbf{n} \in  \mathcal{N}_{\mathcal{M}}$.  By
  explicit computation one has
  \begin{align*}
    &  \argmax_{\mathbf{n}   \in  \mathcal{N}_{\mathcal{M}}}
    \left| \mathbf{v}  \left( E^{\gamma}_{\boldsymbol{\rho}}
    \left( \mathbf{n}^*  \right) \right) \right|_2^2 \\  = &
    \argmax_{\mathbf{n}  \in \mathcal{N}_{\mathcal{M}}}  \Tr
    \left[\gram  \left(  \gamma   -  \overline{\gamma}  +  k
      \right) X_{\mathbf{n}}  \gram \left(  \mathbf{v} \circ
      \boldsymbol{\rho} \right) X_{\mathbf{n}}^T \right].
  \end{align*}
  Take $k$ such that $\gamma - \overline{\gamma} + k \ge 0$.
  Since  $\gram (  \gamma -  \overline{\gamma} +  k )$  is a
  product  matrix and  $\gamma -  \overline{\gamma} +  k$ is
  non-negative  and   non-decreasing,  $\gram  (   \gamma  -
  \overline{\gamma}  +  k  )$ is  monotone  anti-Monge  (see
  Ref.~\cite{BCRW98}  for   the  definition   of  anti-Monge
  matrix).     Hence,    the    statement    follows    from
  Corollary~\ref{cor:qubit}   and    from   Theorem~1.6   of
  Ref.~\cite{BCRW98}.   The case  in which  $\gamma$ is  not
  non-decreasing   or   $-    \gram   (   \mathbf{v}   \circ
  \boldsymbol{\rho}   )$  is   not  benevolent   follows  by
  Lemma~\ref{lmm:covariance}.
\end{proof}

\section{Examples}

For any  set $\mathcal{M}  = \{1, \dots,  |\mathcal{M}| \}$,
any two-dimensional Hilbert space  $\mathcal{H}$, and any $h
\ge  0$, let  ensemble $\boldsymbol{\rho}_h  \in \mathcal{E}
(\mathcal{M}, \mathcal{H})$ be such that
\begin{align*}
  \mathbf{v}  \left( \boldsymbol{\rho}_h  \left( m  \oplus 1
  \right) \right) = \lambda \left( \cos \frac{2 \pi \left( m
    \right)}{\left| \mathcal{M}  \right|}, \sin  \frac{2 \pi
    \left(    m   \right)}{\left|    \mathcal{M}   \right|},
  \left(-1\right)^m h \right),
\end{align*}
for any  $m \in \mathcal{M}$  and for some $\lambda  \ge 0$.
Notice that  for $h  = 0$  the image  $\boldsymbol{\rho}_h (
\mathcal{M} )$ corresponds to a regular polygon in the Bloch
sphere,   whereas  for   $|\mathcal{M}|$   even  the   image
$\boldsymbol{\rho}_h  (  \mathcal{M}  )$ corresponds  to  an
anti-prism in the Bloch sphere.

The   following   corollary   proves  the   benevolence   of
regular-polygonal and anti-prismatic qubit ensembles.

\begin{cor}
  \label{cor:antiprism}
  For any set $\mathcal{M}  = \{1, \dots, |\mathcal{M}| \}$,
  any  non-decreasing  balanced  function $\gamma  :  \{  1,
  \dots,    |\mathcal{M}|    \}   \to    \mathbb{R}$,    any
  two-dimensional  Hilbert space  $\mathcal{H}$, and  any $h
  \ge 0$ such that
  \begin{align}
    \label{eq:antiprism}
    h \le \begin{cases} 0  & \textrm{if } \left| \mathcal{M}
      \right| \textrm{ odd},\\  \sqrt{\frac{1 - \cos \frac{2
            \pi}{   \left|   \mathcal{M}  \right|   }}2}   &
      \textrm{if   }   \frac{\left|  \mathcal{M}   \right|}2
      \textrm{   even},\\  \sqrt{\frac{\cos   \frac{2  \pi}{
            \left|  \mathcal{M}  \right|  } -  \cos  \frac{4
            \pi}{   \left|   \mathcal{M}  \right|   }}2}   &
      \textrm{if   }   \frac{\left|  \mathcal{M}   \right|}2
      \textrm{ odd},
    \end{cases}
  \end{align}
  the                                            measurement
  $\boldsymbol{\pi}^{\gamma}_{\boldsymbol{\rho}_h,
    \mathbf{n}_{\bv}}$   attains   the   minimum   guesswork
  $G^{\gamma}_{\min}  (  \boldsymbol{\rho}_h   )$,  that  is
  $G^{\gamma}_{\min} ( \boldsymbol{\rho}_h  ) = G^{\gamma} (
  \boldsymbol{\rho}_h,
  \boldsymbol{\pi}^{\gamma}_{\boldsymbol{\rho}_h,
    \mathbf{n}_{\bv}} )$.
\end{cor}

\begin{proof}
  If $h =  0$, by the vertex invariance  of regular polygons
  it immediately  follows that $\overline{\boldsymbol{\rho}}
  =   \openone   /   2$  and   $-\gram(   \mathbf{v}   \circ
  \boldsymbol{\rho} )$  is a symmetric Toeplitz  matrix that
  satisfies  the inequality  in Property~\ref{item:monotone}
  and        saturates        the       inequality        in
  Property~\ref{item:majorize},  hence the  first inequality
  in       Eq.~\eqref{eq:antiprism}       follows       from
  Corollary~\ref{cor:benevolent}.

  If $| \mathcal{M} |$ is  even, by the vertex invariance of
  anti-prisms      it      immediately     follows      that
  $\overline{\boldsymbol{\rho}} = \openone / 2$ and $-\gram(
  \mathbf{v}  \circ  \boldsymbol{\rho}  )$  is  a  symmetric
  Toeplitz   matrix  that   saturates   the  inequality   in
  Property~\ref{item:majorize},  hence  in  order  to  apply
  Corollary~\ref{cor:benevolent}    we    need   to    prove
  Property~\ref{item:monotone}.  By explicit computation one
  has
  \begin{align*}
    -\gram   \left(   \mathbf{v}   \circ   \boldsymbol{\rho}
    \right)_{m \oplus 1, 1}  = \lambda^2 \left( \cos \frac{2
      \pi  m}{  \left|  \mathcal{M}  \right|}  +  \left(  -1
    \right)^m h^2 \right).
  \end{align*}
  Two cases need be distinguished.  If $| \mathcal{M} | / 2$
  is  even, by  explicit computation  $- \gram  ( \mathbf{v}
  \circ    \boldsymbol{\rho})_{m    \oplus   1,    1}$    is
  non-increasing in $m$ if and only if
  \begin{align*}
    -\gram   \left(   \mathbf{v}   \circ   \boldsymbol{\rho}
    \right)_{\left|  \mathcal{M} \right|  /  2,  1} +  \gram
    \left(      \mathbf{v}      \circ      \boldsymbol{\rho}
    \right)_{\left| \mathcal{M} \right| / 2 - 1, 1} \le 0,
  \end{align*}
  that is
  \begin{align*}
    -1  - \cos  \left(  \frac{\left|  \mathcal{M} \right|  -
      2}{\left| \mathcal{M} \right|} \pi \right) + 2 h^2 \le
    0,
  \end{align*}
  which   is  equivalent   to  the   second  inequality   in
  Eq.~\eqref{eq:antiprism}. If $| \mathcal{M} | / 2$ is odd,
  by  explicit  computation  $-  \gram  (  \mathbf{v}  \circ
  \boldsymbol{\rho})_{m \oplus  1, 1}$ is  non-increasing in
  $m$ if and only if
  \begin{align*}
    -\gram   \left(   \mathbf{v}   \circ   \boldsymbol{\rho}
    \right)_{\left| \mathcal{M} \right| / 2  - 1, 1} + \gram
    \left(      \mathbf{v}      \circ      \boldsymbol{\rho}
    \right)_{\left| \mathcal{M} \right| / 2 - 2, 1} \le 0,
  \end{align*}
  that is
  \begin{align*}
    \cos \left( \frac{\left| \mathcal{M} \right| - 2}{\left|
      \mathcal{M}  \right|}   \pi  \right)  -   \cos  \left(
    \frac{\left| \mathcal{M} \right| - 4}{\left| \mathcal{M}
      \right|} \pi \right) + 2 h^2 \le 0,
  \end{align*}
  which   is  equivalent   to   the   third  inequality   in
  Eq.~\eqref{eq:antiprism}.
\end{proof}

Notice  that the  real and  complex SIC  and MUBs  ensembles
fulfill  the  hypothesis  of  Corollary~\ref{cor:antiprism}.
Indeed,  the  real  SIC  and MUBs  ensembles  correspond  to
regular polygons  in the  Bloch sphere  with three  and four
vertices, respectively;  that is,  they are obtained  for $|
\mathcal{M} | = 3$ and  $| \mathcal{M} | = 4$, respectively,
with  $h  =  0$.   Analogously, the  complex  SIC  and  MUBs
correspond  to the  regular  simplex and  octahedron in  the
Bloch sphere,  which in turn  are equivalent to  the uniform
(and thus  regular) anti-prisms with four  and six vertices;
that is, they are obtained for  $| \mathcal{M} | = 4$ and $|
\mathcal{M} |  = 6$,  respectively, with $h$  saturating the
bounds   in   Eq.~\eqref{eq:antiprism},   that   is   $h   =
1/\sqrt{2}$.

\section{Conclusion}

In  this  work,  for  arbitrary balanced  cost  function  we
derived  (Theorem~\ref{thm:helstrom}) sufficient  conditions
under  which  the guesswork  is  attained  by a  two-outcome
measurement,   as   well   as   the   closed-form   solution
(Corollary~\ref{cor:qubit})  of  the guesswork  problem  for
qubit  ensembles with  uniform  probability distribution  in
terms  of  the quadratic  assignment  problem,  that is,  an
optimization    over    a    finite    set.     We    solved
(Corollary~\ref{cor:benevolent})  the  quadratic  assignment
problem  for  a class  of  qubit  ensembles referred  to  as
benevolent,  and  we proved  (Corollary~\ref{cor:antiprism})
the benevolence of  anti-prismatic ensembles, including real
and complex SICs and MUBs.

\section{Acknowledgments}

The  author is  grateful to  Alessandro Bisio  and Francesco
Buscemi  for insightful  comments.  The  author acknowledges
support  from   the  Department  of  Computer   Science  and
Engineering,   Toyohashi  University   of  Technology,   the
International  Research Unit  of Quantum  Information, Kyoto
University, and the JSPS KAKENHI Grant Number JP20K03774.

\textbf{Michele Dall'Arno}  received his PhD  in theoretical
physics from  the University  of Pavia,  Italy, in  2012. He
held  positions  in   ICFO  (Barcelona),  Nagoya  University
(Japan),  the  National   University  of  Singapore,  Waseda
University  (Japan), and  Kyoto  University (Japan).   Since
2023, he  is associate  professor in quantum  information in
Toyohashi  University  of  Technology (Japan)  and  visiting
researcher in Kyoto University (Japan).
\end{document}